\documentclass{article}
\usepackage{amsmath}
\usepackage{amsfonts,amssymb}
\usepackage{eucal}
\usepackage{amsthm}
\usepackage{graphicx}
\usepackage{color} 

\newcommand{\simgt}{\lower.5ex\hbox{$\; \buildrel > \over \sim \;$}}
\newcommand{\simlt}{\lower.5ex\hbox{$\; \buildrel < \over \sim \;$}}

\newcommand{\Mod}{\text{\rm mod }}

\newcommand{\QR}{\text{\rm QR}}
\newcommand{\BYTEMIX}{\text{\rm BYTE\_MIX}}

\newtheorem{thm}{Theorem}

\newtheorem{rem}{Remark}

\begin{document}
%
\title{Crashing Modulus Attack on Modular Squaring for Rabin Cryptosystem}

\author{Masahiro~Kaminaga, 
        Hideki~Yoshikawa,
        Arimitsu Shikoda,
        and~Toshinori~Suzuki
\thanks{M. Kaminaga, H. Yoshikawa, and T. Suzuki are with the Department 
of Electrical Engineering and Information Technology, and A. Shikoda is with the Department 
of Electrical Engineering, Tohoku Gakuin University, 
13-1, Chuo-1, Tagajo, 985-8537, Japan e-mail: (kaminaga@mail.tohoku-gakuin.ac.jp).}
\thanks{Manuscript received xxxx xx, xxxx; revised xxxx xx, xxxx.}}

\markboth{IEEE Transactions on Information Forensics and Security,~Vol.~xx, No.~x, xx~20xx}%
{Kaminaga \MakeLowercase{\textit{et al.}}: Crashing Modulus Attack on Modular Squaring}

\maketitle

\begin{abstract} 
The Rabin cryptosystem has been proposed protect the unique ID (UID) in
radio-frequency identification tags. 
The Rabin cryptosystem is a type of lightweight public key system that
is theoretetically quite secure; however it is vulnerable 
to several side-channel attacks. 
In this paper, a crashing modulus attack is presented as a new fault attack on
 modular squaring during Rabin encryption.
This attack requires only one fault in the public key if its perturbed public key can be factored. 
Our simulation results indicate that the attack is more than 50\% successful with several faults in practical time.
A complicated situation arises when reconstrucing the message, including the UID, from ciphertext, i.e., 
the message and the perturbed public key are not relatively prime. 
We present a complete and mathematically rigorous 
message reconstruction algorithm for such a case.
Moreover, we propose an exact formula to obtain a number of candidate messages.
We show that the number is not generally equal to a power of two.
\end{abstract}


%

\section{Introduction}\label{intro}
It is quite important that 
the modular squaring map $f(x) = x^2~\bmod N$ from $\mathbb{Z}_N^*=(\mathbb{Z}/N\mathbb{Z})^*$ to itself, 
where $\mathbb{Z}_N^*=(\mathbb{Z}/N\mathbb{Z})^*$, is the reduced residue system mod $N$.
The algorithmic complexity of 
the modular quadratic equation $x^2\equiv C\pmod{N}$ for quadratic residue $C\bmod N$ is equivalent to factoring $N$. 
Thus, $f$ has one-wayness if factoring $N$ is sufficiently difficult.

There are various applications of the modular squaring map with $N=pq$ for large distinct primes $p$ and $q$.
For example, Blum et al.~\cite{BBS} constructed a unpredictable pseudorandom sequence to take the parity of $x_i$ for 
$x_i = f(x_{i-1})~\Mod N$ with a secret initial value $x_0\in\mathbb{Z}_N^*$. 
This is known as the Blum-Blum-Shub pseudorandom number generator, and it is theoretically quite important. 
Another important example is the Rabin cryptosystem developed by Michael C. Rabin~\cite{Rabin}.
The Rabin cryptosystem is a public key cryptosystem with public key $N$ and secret keys $p$ and $q$. 
In this system, message $M$ is encrypted as $C = M^2~\bmod N$ and decrypted as 
four possible roots of $C$ using $\sqrt{C}~\bmod p$ and $\sqrt{C}~\bmod q$ and the 
Chinese remainder theorem~(CRT)~(see Section \ref{Rabin}). 

The security of the Rabin cryptosystem, similar to that of RSA, 
is related to the practical difficulty of factoring $N=pq$.
The Rabin cryptosystem has a theoretical advantage in that there exists an exact proof of 
its security equivalent to factoring $N$, which is not currently known to be true for RSA.  
Being theoretically important, the Rabin cryptosystem is also useful 
for passive radio frequency identification~(RFID). 
The use of passive RFID tags to prevent counterfeiting by embedding them in a product is an emerging application. 
RFID systems comprises tags and interrogators.
RFID tags are low-cost wireless devices that associate a unique ID~(UID) with the product.
These tags are powered passively by the interrogator.
Implementation of a public key cryptosystem on RFID tags is challenging, 
because the hardware limited.  Therefore, ``lightweight" cryptosystems are required for RFID tags.
RSA is a well-known and effective public key cryptosystem; however, 
it is not suitable for RFID tags.
RSA encryption requires many modular multiplications, relatively long processing time, 
and a relatively large data-path area. 
In addition, RSA consumes a significant amount of energy. 
There are two major lightweight public key cryptosystems suitable for RFID tags, i.e., 
the elliptic curve cryptosystem~(ECC) and the Rabin cryptosystem. 

ECC can be applied to small devices and has shorter processing time than RSA. 
Moreover, ECC is suitable for various RFID applications. 
Many academic papers on ECC for RFID tags have been published.
For example, F\"urbass-Wolkerstorfer~\cite{Furbass}, Lee-Sakiyama-Batina-Verbauwhede~\cite{Lee}, 
Hutter-Feldhofer-Plos~\cite{Hutter1}, Hutter-Feldhofer-Wolkerstorfer~\cite{Hutter2}, 
Hutter-Joye-Sierra~\cite{Hutter3}, Pessl-Hutter~\cite{Hutter4}, and Kern-Feldhofer~\cite{Kern} 
have reported remarkable results on ECC implementation for RFID tag chip.

The Rabin cryptosystem requires only one modular squaring, which 
is advantageous for use with RFID tags. 
Indeed, Oren-Feldhofer~\cite{Oren-Feldhofer}, \cite{Oren-Feldhofer2}, Arbit et al.~\cite{Implement} 
successfully applied Rabin encryption variant to RFID tags. The variant, known as 
WIPR~(Weizmann-IAIK[Institute for Applied Information Processing and Communications] Public Key for RFID) 
was developed by Naccache~\cite{Naccache} and Shamir~\cite{Shamir}, \cite{Shamir2}. 
WIPR is smaller, faster, and requires less power than ECC implementations. 
Giesecke \& Devrient GmbH~\cite{RAMON} proposed 
proposed the Rabin-Montgomery Cryptosystem~(RAMON), 
a public key protocol for RFID tags based on the Rabin cryptosystem. 
RAMON uses Montgomery reduction~\cite{Mont} to avoid trial division. 
Therefore, it is very likely that the Rabin cryptosystem 
will be implemented on  various types of RFID tag chips. 

Since the publication of Boneh, DeMillo, and Lipton's landmark paper~\cite{Bellcore}, 
differential fault analysis(DFA) has been an active area in cryptography. 
DFA is a technique to extract secret information from a cryptographic device by provoking a computational fault. 
DFA is a real threat for cryptographic devices, such as smartcard~\cite{Bar}; therefore, 
a vast number of research papers about DFA have been published. 
The monograph edited by Joye-Tunstall~\cite{Joye-Tunstall-Book} is a good guide to this field. 

However, conventional fault attack research for public key cryptosystems 
has focused on DFA for smartcards, particularly signature schemes using RSA and ECC. 
Little attention has been paid to Rabin cryptosystem implemented on an RFID tag chip. 
We propose a powerful fault attack by one-byte perturbation of public key $N$
 based on the assumption that an attacker can induce faults as 
 the device moves one byte of $N$ from non-volatile memory to a register.
Under this assumption, the attacker can create a new faulted public key $\hat{N}=p_1^{e_1}p_2^{e_2}\cdots p_{\omega}^{e_{\omega}}$, 
where $p_j$ are mutually distinct primes and $e_j$ are positive integers. 

We provide a mathematical analysis and 
demonstrate the effectiveness of the proposed fault attack through simulation.
Although there have been some related studies~\cite{Seifert}, \cite{Muir}, they 
are not directly applicable to our target.
An attack against RSA, first developed by Seifert~\cite{Seifert} and
extended to the general case by Muir et al.~\cite{Muir}, can obtain a new secret
exponent $e^{-1}\bmod \varphi(\hat{N})$, where $e$ is a public exponent and 
$\varphi$ is Euler's totient function. 
This approach is not applicable to Rabin cryptosystems; 
therefore, a secret message $M$ including UID must be reconstructed directly.
Conversely, it is known that an attacker can obtain modular quadratic 
equation $M^2 \equiv C\pmod{\hat{N}}$ by solving each $M^2\equiv C\pmod{p_j^{e_j}}$ and
computing $M$ using CRT only if $\gcd(M,\hat{N})=1$ holds.
However, generally, 
some prime factors of $\hat{N}$ are small. 
As a result, cases in which $\gcd(M,\hat{N})\ne 1$ occur frequently. 
We also present a complete mathematical method to reconstruct the message in such cases. 

\medskip
The remainder of this paper is organized as follows. 
Section II presents a brief description of quadratic residues, 
the Rabin cryptosystem, and basic facts about the target implementation of a Rabin cryptosystem 
in an RFID tag.
Section III presents the general principle and procedure of our attack, 
as well as a complete and mathematically rigorous message reconstruction algorithm.
We show an exact formula to obtain the number of candidate message in Section IV. 
Simulated attack results are presented in Section V, and conclusions are presented in Section VI. 
\section{Preliminaries}
\subsection{Quadratic Residues}
Let $p$ be an odd prime. 
An integer $C\in\mathbb{Z}_p^*$ is called a quadratic residue mod $p$ if 
there exists $x$ such that $x^2\equiv C\pmod{p}$. 
We denote the set of all quadratic residues mod $p$ by $\QR_p$.
We can use 
Euler's criterion to claim that $C^{\frac{p-1}{2}}\equiv 1\pmod{p}$ 
to determine $C\in\QR_p$ or not.
In other words, $C^{\frac{p-1}{2}}\equiv 1\pmod{p}$ holds if $C\in\QR_p$, or 
$C^{\frac{p-1}{2}}\equiv -1\pmod{p}$ holds if $C\not\in\QR_p$.
Using Euler's criterion, square roots of $C$ mod $p$ can be represented as 
$$
\sqrt{C}=\pm C^{\frac{p+1}{4}}\bmod p
$$
if $p\equiv 3\pmod{4}$. 
Note that $\frac{p+1}{4}$ is a positive integer under the condition 
and $(\pm C^{\frac{p+1}{4}})^2\equiv C^{\frac{p+1}{2}}\equiv C^{\frac{p-1}{2}}C\equiv C\pmod{p}$.
If $p\not\equiv 3\pmod{4}$, i.e., $p\equiv 1\pmod{4}$, 
then we require the Tonelli-Shanks algorithm to find the square roots. 
The Tonelli-Shanks algorithm runs in polynomial time 
assuming that the generalized Riemann hypothesis is true. 
We can find the roots mod $p^e$ of $C$ as follows. 
First, we find all square roots mod $p$. 
Then using the Hensel lift(see e.g., section 13.3.2 of Shoup's book~\cite{Shoup}), we 
lift each of these square roots to obtain all of the roots $p^2$, 
and then lift these to obtain all square roots 
mod $p^3$, and so on.
Quadratic residue can be generalized for the mod of 
an odd composite number $N=p_1^{e_1}p_2^{e_2}\cdots p_{\omega}^{e_{\omega}}$. 
A quadratic residue mod $N$ is an integer $\beta$ such that there exists an $\alpha\in\mathbb{Z}_N^*$ 
that satisfies $\beta\equiv \alpha^2\pmod{N}$. 
Using the CRT, it can be demonstrated that $\beta$ is a quadratic residue mod $N$ 
if and only if it is a quadratic residue mod of each $p_i^{e_i}$, and one obtain every root mod $N$ 
from all roots mod $p_i^{e_i}$'s.

\subsection{Rabin Cryptosystem}\label{Rabin}
The Rabin cryptosystem~\cite{Rabin} is a public key system based on the factorization 
difficulty of $N=pq$ where $p$ and $q$ are large and distinct balanced primes. 
The length $n$ of $N$ must be greater than or equal to 1,024 to be safe.
$N$ is its public key, and $p$ and $q$ are its secret keys. 
To reduce decryption complexity, choose $p$ and $q$ that satisfy 
$p\equiv q\equiv 3\pmod{4}$ should be chosen. 
According to the Dirichlet's theorem on arithmetic progressions
(see e.g., Theorem 5.52 of Shoup's book~\cite{Shoup}), 
infinitely many prime numbers $p$ that satisfy $p\equiv 3\pmod{4}$ exist. 
Let $\mathbb{Z}_N^*=(\mathbb{Z}/N\mathbb{Z})^*$ be the reduced residue system mod $N$.
Generally, the plaintext $M\in\mathbb{Z}_N^*$ is generated from a shorter message, including 
the UID, in our case. 
In the Rabin cryptosystem, to encrypt $M\in\mathbb{Z}_N^*$, the sender computes its square mod $N$:
$$
C = M^2 \bmod N.
$$ 

To decrypt the ciphertext $C$, the receiver computes its square roots $\sqrt{C}$ 
of $C$ mod $N$ using $p$ and $q$ as follows.
First, compute 
$C_p = \sqrt{C}\bmod p = C^{\frac{p+1}{4}}\bmod p$ and 
$C_q = \sqrt{C}\bmod q = C^{\frac{q+1}{4}}\bmod q$
using an efficient exponentiation algorithm.
Second, using the CRT, 
the four roots are computed as follows:
$$
\sqrt{C}=
\pm C_p q(q^{-1}\bmod p)\pm C_q p(p^{-1}\bmod q).
$$
Finally, 
the receiver recognizes the valid plaintext based on its format, such as redundancy and structure.

The Rabin cryptosystem has two significant advantages with respect to 
alternative public key schemes. 
First, it is provably difficult to factor $N$. 
Second, it imposes a small computational burden, 
has relatively 
lightweight implementation, and requires only a single squaring and modular reduction 
for encryption.

lightweight implementation, and requires

%
%

\subsection{WIPR Scheme}
The most time consuming process of Rabin encryption is trial division by $N$ 
because it is a RAM-intensive process. 
There are two well-known ways to avoid trial division. 
One is using Montgomery reduction~\cite{Mont}, 
and the other is using the WIPR scheme.
When we use Montgomery reduction, we 
compute $A^2R^{-1} \bmod N(R=2^n)$ rather than $A^2 \bmod N$, 
where $N$ is $n$-bits long; therefore, it is simply a data format problem. 
Conversely, the WIPR scheme includes an essentially different process.
Thus, we describe only the WIPR scheme.

%
To reduce the trial division process, Naccache~\cite{Naccache} and Shamir~\cite{Shamir}, \cite{Shamir2} 
proposed a variant by replacing the modular multiplication by adding a large random multiple of $N$, 
where the size of the random number $r$ is at least 80 bits longer than the size of $N$:
$$
C' = M^2 + rN.
$$
Obviously, $C'\equiv C\pmod{N}$, and $C'$ is fully randomized.
The decryption process is identical to Rabin's original process.
This randomized variant of Rabin's scheme is easier to implement because it has only multiplications 
without modular reduction. 
It is lighter than the original Rabin scheme; however, it 
requires a register that is approximately twice as long 
for the ciphertext.

The WIPR scheme replaces $r$ with the output of a light stream cipher, 
which was developed by Oren-Feldhofer~\cite{Oren-Feldhofer}. 
This stream cipher is implemented by creating a Feistel network. 
Arbit et al.~\cite{Implement} reported that 
their successful implementation had a data-path area of 
4,184 gate equivalents, an encryption time of 180 ms and an average power consumption of 11 $\mu$W.  

We describe the WIPR challenge-response protocol as follows.
\begin{enumerate}
\item Challenge: The interrogator sends the challenge(random bit string) $c$ 
of length $s$ to the tag.
\item Response: 
The RFID tag generates two random bit strings $R_{\mbox{\small tag}}$ and $r$, where 
$|R_{\mbox{\small tag}}| = n - s - |\mbox{UID}|$, and $|r|=n + t$.
The tag generates a message as follows:
$$
M = \BYTEMIX(c || R_{\mbox{\small tag}} || \mbox{UID}),
$$
where $||$ denotes concatenation operator, and transmits the following ciphertext:
$$
C' = M^2 + rN, 
$$
and {\BYTEMIX} is a simple byte-interleaving 
operation~(see Oren-Feldhofer~\cite{Oren-Feldhofer2} for details).
\item Verify: The interrogator decrypts $C'$ using the secret key $(p,q)$ and 
finds the correct message, including the UID, in four square roots.
\end{enumerate}
Here, $s$ and $t$ are security parameters~(originally set to $s=t=80$).

Note that using the Rabin function $f(x)=x^2\bmod N$ to encrypt a message $M$ 
that satisfies $|M|<n$ requires 
some kind of random padding. 
Some padding schemes with short random padding are vulnerable to attacks based on 
Coppersmith's Theorem for a univariate polynomial~\cite{Coppersmith} and 
Franklin-Reiter's related message attack~\cite{FR1}, \cite{FR2}.

\subsection{RAMON}
Another way to avoid trial division is using 
Montgomery reduction~\cite{Mont}.
When we use Montgomery reduction, we 
compute $ABR^{-1} \bmod N(R=2^n)$ rather than $AB \bmod N$, 
where $N$ is $n$-bits long.
Montgomery reduction computes $S$ that satisfies the following Diophantine equation:
$$
AB + TN = SR.
$$
Clearly, $S\equiv ABR^{-1} \pmod{N}$. 

RAMON was proposed by Giesecke \& Devrient GmbH~\cite{RAMON}.
RFID tag sends the following ciphertext:
$$
C^* = M^2R^{-1}\bmod N,
$$
where $R=2^n$.
In 1024 bit $N$ case, the message $M$ is formatted as follows:
$M$(128 bytes) = challenge(10 bytes) $||$ tag random number(10 bytes)
 $||$ TLV-coded signed Tag UID(n bytes) $||$ variable length filling(x bytes) $||$ 
checksum(2 bytes) $||$ The last byte must be left free; i.e., set to zero~(1 byte).

Note that $C^*\ne C$; therefore, 
the interrogator transforms $C^*$ into normal ciphertext as follows:
$$
C = C^*R\bmod N = (M^2R^{-1})R\bmod N.
$$
Then, the interrogator computes four roots of $C$ mod $N$ using the secret key $(p,q)$ 
and finds the correct $M$ based on its format.

%
%

\section{Proposed Attack}

\subsection{Principle}
In the following, we consider only the WIPR protocol for convenience. 
The fundamental idea of our attack method uses perturbed public key $\hat{N}$ of $N$.
In this case, the ciphertext changes as follows:
\begin{equation}\label{perturbed_eq}
\hat{C} = M^2 + r\hat{N}.
\end{equation}
Generally, the attacker cannot factor the coprime $N$ in realistic time.
Conversely, a perturbation $\hat{N}$ of $N$ can be factored at high probability. 
If the attacker has factored $\hat{N}$ successfully, such as 
$\hat{N}=p_1^{e_1}p_2^{e_2}\cdots p_{\omega}^{e_{\omega}}$, 
then the modular quadratic equation 
$\hat{C} = M^2 \pmod{\hat{N}}$ derived from (\ref{perturbed_eq}) splits into $\omega$ smaller 
equations: 
\begin{equation}\label{perturbed_eq_splitted}
 M^2\equiv \hat{C} \pmod{p_j^{e_j}}, \quad j=1,2,\cdots, \omega.
\end{equation}
Equation (\ref{perturbed_eq_splitted}) can be solved using the 
Tonelli-Shanks algorithm and the Hensel lift. 
CRT leads us to all roots of (\ref{perturbed_eq}) 
from the roots $M_{jk}(k\geq 2)$ of (\ref{perturbed_eq_splitted}), i.e., 
we obtain:
\begin{equation}\label{CRT}
M_0 = \sum_{j=1}^{\omega}M_{jk}\hat{N}_j(\hat{N}_j^{-1}\bmod p_j^{e_j}),
\end{equation}
where $\hat{N}_j=\hat{N}/p_j^{e_j}$. 
The attacker obtains the correct roots $M$ by modifying the above $M_0$ as $(M_0+k\hat{N})\bmod{N}$ for 
the smallest $k$ such that $M_0+k\hat{N}>N$. 
The perturbed modular quadratic equation (\ref{perturbed_eq}) typically has $2^{\omega}$ roots, and 
these roots contain the correct message including the UID.
The number of roots exceeds $2^{\omega}$ in some cases depending 
on the values of $\gcd(\hat{C},p_j^{e_j})$ for $j=1,2,\cdots,\omega$. 
We discuss these problems in Sections \ref{roots} and \ref{NUM}.

\subsection{Fault Models}
The WIPR protocol requires two online multiplications to compute $C = M^2 + rN$. 
Optimal implementation of WIPR with 1,024 bit $N$ was shown by 
Arbit et al.~\cite{Implement}.
This multiplication process is performed on a multiply-accumulate register by convolution. 
Assuming a word size of one byte, a single multiply-accumulate register 
perform this multiplication in approximately $2^{16}$ steps. 
The public key $N$ moves from non-volatile memory to the register byte by byte. 
We assume that the attacker can inject a one-byte fault into this data moving process.
In this paper, we consider two fault models. 
\subsubsection{Crash a byte of $N$}
The first fault model that we choose to perform our attack 
with is derived from those used by Berzati et al.~\cite{Berzati2}, \cite{Berzati3} 
to successfully attack standard RSA.

Here let $\mathbb{Z}[a,b]$ be a set of integers in the interval $[a,b]$. 
We assume that the attacker can inject a transient fault that public key $N$ modifies by byte, that is,
the injected fault affects only one byte of the public key by modifying it randomly as follows:
$$
\hat{N} = N \oplus \epsilon
$$
where $\oplus$ is bitwise exclusive OR and $\epsilon = R_i\cdot 2^{8i}$, $R_i\in\mathbb{Z}[1,2^8-1]$ 
for $i\ne 0$ which 
is required to preserve the parity of $\hat{N}$. 
We assume the attacker knows the position $i$, 
but the correct value of the faulty public key $\hat{N}$ is unknown by the attacker.
The attacker must factor 255~($=2^8-1$) candidates of $\hat{N}$. 
Our attack also works for a fault that affects several bytes of $N$. 
However, the attacker's task grows in proportion to 
the number of candidates $\hat{N}$ of perturbed $N$.

This is a natural assumption for both WIPR and RAMON. 
In the WIPR case, the attack target is the time at 
which the $i$-th byte $N[i]$ of $N$ moves from non-volatile memory to 
the register for multiplication before multiplying $r$ and $N$. 
In the RAMON case, the fault is injected while $N$ moves from 
non-volatile memory, such as EEPROM, 
to a register at the trensfer time of the $i$-th byte 
prior to Montgomery squaring $M^2R^{-1}\bmod N$.

\subsubsection{Instruction skip}
The second fault model is based on the instruction skip technique. 
Instruction skip is equivalent to replacing an instruction with a no operation 
in assembly language.
Several researchers have investigated DFA using an instruction skip, or a bypass operation 
\cite{RoundReduction}, \cite{park}, \cite{yoshikawa1}. 
Instruction skip does not affect the registers, internal memory, and calculation process.
Successful instruction skip attacks have been reported for 
PIC16F877~\cite{RoundReduction}, ATmega 128~\cite{park}, 
and ATmega 168~\cite{yoshikawa1} microcontrollers. 
Choukri-Tunstall~\cite{RoundReduction} and Park et al.~\cite{park} showed that 
an entire Advanced Encryption Standard secret key could be 
reconstructed by skipping a branch instruction used to 
check the number of rounds. 
Kaminaga et al.~\cite{KYS2015IEEE} showed that it is possible to 
reconstruct an entire secret exponent with $63(=2^6-1)$ 
faulted signatures in a short time for a 1536-bit RSA implementation with the 
$2^6$-ary method using instruction skipping technique in precomputation phase.

%
%

Our attack target is a conditional branch operation for moving 
the last byte of $N$ at the counter $i=127$.
If the conditional branch operation is 
skipped, the attacker obtains the faulted public key $\hat{N}$ as follows:
$$
\hat{N} = \sum_{i=0}^{126}N[i](2^8)^i,
$$
where each $N[i]\in \mathbb{Z}[0,255]$.
Clearly, $\hat{N}$ is one byte shorter than the original $N$, and preserves its parity.

\subsection{Target Byte Location}
Some Rabin cryptosystems for RFID tags adopt special types of modulus $N$ for restricted hardware resources.
WIPR for RFID proposed by Oren-Feldhofer~\cite{Oren-Feldhofer2} uses 
modulus $N$ with a predefined upper half to reduce ROM cost by half.
RAMON uses modulus $N$ to satisfy the condition $N\equiv 1\pmod{2^{n/2}}$, which means 
that approximately one-half of the least significant bits of $N$~(except for the last one) are zeroes 
to reduce multiplications~\cite{RAMON}.
Half of the processes in which modulus data is transferring from EEPROM to a register are reduced in 
such cases.
Therefore, the attacker must set the location of the target byte to the lower half bytes of $N$ 
for implementation of 
Oren-Feldhofer's proposal~\cite{Oren-Feldhofer2}, and the attack must set the 
location of the byte to upper half bytes of $N$ for RAMON~\cite{RAMON}.

\subsection{Attack Procedure}\label{AP}
The attacker's goal is to reconstruct the secret message $M$, including the UID.
The following steps provide an example of our attack process. 
\begin{description}
\item[Step 1.] ~Create a perturbed public key $\hat{N}$ 
by injecting a fault to a byte of the public key $N$.
\item[Step 2.] ~Factorize 255 candidates $\hat{N}_k(k=1,2,\cdots, 255)$ of $\hat{N}$. 
When factoring consumes too much amount of time, perturb another byte of $N$ and attempt this process again.
\item[Step 3.] ~Solve modular quadratic equation $x^2\equiv \hat{C}\pmod{\hat{N}_k}$.
\item[Step 4.] ~Find the correct message $M$ based on data format in all roots of $x^2\equiv \hat{C}\pmod{\hat{N}_k}$.
\end{description}
Step 2 is the most time concuming process in computation for the attacker. 
Most of $\hat{N}_k(k=1,2,\cdots, 255)$ have relatively small factors; thus, these 
can be factored in a short time.
However, some cases require more time to factor $\hat{N}_k$. Then, the attacker 
shifts the position of the target byte and attempt the factoring process again. 
Step 3 is a technical process, and we must consider the degenerate case, 
$\gcd(\hat{C},\hat{N}_k)\ne 1$.
Mathematicians have paid little attention to the degenerate case; however, 
such case arises in our attack. 

\subsection{Reconstruction of Roots}\label{roots}
Our attack method comes down to finding all roots of the following 
modular quadratic equation with a square number $\hat{C}$:
\begin{equation}\label{QE}
x^2 \equiv \hat{C}\pmod{p^e},
\end{equation}
where $p$ is an odd prime and $e$ is a positive integer. 
The oddness of $p$ obeys the fact that our attack targets a byte of $N$ that is not the lowest byte.

We use assumption (A), i.e., $\hat{C}$ is squared $\bmod{~p^e}$, throughout this paper. 
Assumption (A) is very natural because $\hat{C}\equiv M^2\pmod{p_j^{e_j}}$ is derived from 
our target equation $\hat{C} = M^2 + r\hat{N}$. 

We denote $\hat{C}\bmod\hat{N}$ by $\hat{C}$. 
The algorithm for finding the roots of (\ref{QE}) depends on whether $\gcd(p,\hat{C})=1$. 
We distinguish the ``degenerate" case $\gcd(\hat{C},\hat{N})\ne 1$ 
from the ``non-degenerate" case $\gcd(\hat{C},\hat{N})= 1$.

\subsubsection{Non-degenerate Case}

Here, let $\hat{C}\equiv M^2\pmod{\hat{N}}$ for some $M\in\mathbb{Z}$.
For $\gcd(p,\hat{C})=1$, finding the roots of (\ref{QE}) 
is not difficult~(see e.g., Section 2.8.2 of Shoup's book~\cite{Shoup_book}). 
\begin{thm}(Theorem 2.25~\cite{Shoup_book})\label{justtwo}
If $\gcd(p,\hat{C})=1$, then the modular quadratic equation (\ref{QE}) 
is equivalent to $x\equiv \pm M\pmod{p^e}$. 
In particular, (\ref{QE}) has only two roots.
\end{thm}
\begin{proof}
Here, we show the proof for Theorem~\ref{justtwo}. 
Since $\gcd(M,p)=1$, there exists the inverse $M^{-1}$ mod $p$. 
Therefore, (\ref{QE}) is equivalent to $(xM^{-1})^2 \equiv 1\pmod{p^e}$.
Let $\gamma$ be $xM^{-1}$, and we obtain
$$
(\gamma+1)(\gamma-1)\equiv 0\pmod{p^e}.
$$
Thus, there exists non-negative integers $\delta_1, \delta_2(\delta_1+\delta_2=e)$ such that 
$p^{\delta_1}|(\gamma+1)$ or $p^{\delta_2}|(\gamma-1)$. 
Then, $2=(\gamma+1)-(\gamma-1)$ divides $p$ if both $\delta_1$ and $\delta_2$ are positive, 
which leads to a contradiction because $p$ is an odd prime, therefore, 
$\delta_1=0$ or $\delta_2=0$.
\end{proof}

\begin{thm}
For $p\equiv 3\pmod{4}$ and $\gcd(p,\hat{C})=1$, all the roots of (\ref{QE}) are given by:
\begin{equation}\label{easyroot}
\sqrt{\hat{C}} \equiv \pm \hat{C}^{\frac{\varphi(p^e)+2}{4}}\pmod{p^e},
\end{equation}
where $\varphi$ is Euler's totient function.
\end{thm}
\begin{proof}
It is easy to verify (\ref{easyroot}) directly. 
From Theorem~\ref{justtwo}, the number of roots of (\ref{QE}) is two. 
Note that $\varphi(p^e)+2=p^{e-1}(p-1)+2\equiv 2(-1)^{e-1}+2\equiv 0\pmod{4}$ follows from 
$p\equiv 3\pmod{4}$.
By using Euler's totient theorem $\hat{C}^{\varphi(p^e)}\equiv 1\pmod{p^e}$ 
and squaring $\hat{C}$, we obtain:
\begin{equation}\label{sqrt_Eu}
\hat{C}^{\varphi(p^e)/2}\equiv 1\pmod{p^e}.
\end{equation}
Using (\ref{sqrt_Eu}), we obtain
$(\hat{C}^{\frac{\varphi(p^e)+2}{4}})^2 \equiv \hat{C}^{\frac{\varphi(p^e)+2}{2}}\pmod{p^e} \equiv \hat{C}\pmod{p^e}$.
This means that $\hat{C}^{\frac{\varphi(p^e)+2}{4}}$ is a square root mod $p^e$ of $\hat{C}$, 
therefore, all square roots mod $p^e$ of $C$ are given by $\pm \hat{C}^{\frac{\varphi(p^e)+2}{4}}$.
\end{proof}

For $p\equiv 1\pmod{4}$, we first solve the following equation:
\begin{equation}\label{diffcase}
x^2 \equiv \hat{C}\pmod{p}.
\end{equation}

We require the Tonneli-Shanks algorithm to solve (\ref{diffcase}). 
The Tonneli-Shanks algorithm can be described as follows.

{\small
\begin{enumerate}
\item Determine $s$ such that $p-1 = 2^st$ where $t$ is odd.
\item Find a non-quadratic residue $v$ mod $p$.
\item Compute $z = \hat{C}^t \bmod p$.
\item Find $u$ such that $(v^t)^u \bmod p = z$.
\item Compute $k = 2^s - u$.
\item Output $\hat{C}^{(t+1)/2}(v^t)^{k/2} \bmod p$ as a solution.
\end{enumerate}
}
It is easy to check if $\hat{C}^{(t+1)/2}(v^t)^{k/2} \bmod p$ is a square root mod $p$ of $\hat{C}$. 
Indeed, using Fermat's little theorem, we obtain 
$(\hat{C}^{(t+1)/2}(v^t)^{k/2})^2 \equiv v^{tu}\hat{C}v^{p-1-tu} \equiv \hat{C}\pmod{p}$.

After solving (\ref{diffcase}), the roots can be lifted to mod $p^u(u>1)$ using the following.

\begin{thm}(Hensel's lifting lemma) 
Let $f(x)\in\mathbb{Z}_p[x]$ and $x_0\in\mathbb{Z}_p$ satisfy 
$f(x_0)\equiv 0\pmod{p^k}$ and $f'(x_0)\not\equiv 0\pmod{p}$.
Then, there is unique $x\in\mathbb{Z}_p$ such that $f(x)=0\pmod{p^{k+m}}$ 
and $x\equiv x_0\pmod{p^k}$. 
Furthermore, this x is unique in mod $p^{k+m}$, and can be represented explicitly as:
$$
x = x_0 + tp^k
$$
where
$$
t = -\frac{f(x_0)}{p^k}\cdot (f'(x_0)^{-1}).
$$
The division by $p^k$ denotes ordinary integer division, 
and the inversion $f'(x_0)^{-1}$ is computed in $\bmod{~p^m}$.
\end{thm}

\subsubsection{Degenerate Case}
Here we consider the degenerate case, i.e., $\gcd(\hat{C},\hat{N})\ne 1$.

\begin{thm}\label{zerothm}
All roots of the modular quadratic equation $x^2\equiv 0\pmod{p^e}$ are given by
$x=kp^{\lceil e/2\rceil}$ for $k\in\mathbb{Z}[0,p^{\lfloor e/2\rfloor}-1]$. 
In particular, the number of its roots is $p^{\lfloor e/2\rfloor}$.
\end{thm}
\begin{proof} It is clear that $x=kp^{\lceil e/2\rceil}$ for $k\in\mathbb{Z}[0,p^{\lfloor e/2\rfloor}-1]$ satisfies 
$x^2\equiv 0\pmod{p^e}$. 
Thus, we only have to show that other forms of the roots do not exist. 
Let $x=kp^t$ with $\gcd(k,p)=1$. Since $\gcd(x^2,p^e)\ne 1$, $x$ is a multiple number of $p$, 
therefore, $t\geq 1$. 
Then $x^2 = k^2p^{2t} \equiv 0\pmod{p^e}$ if and only if $2t\geq e$ holds.
\end{proof}

\begin{thm}\label{nonzerothm}
All roots of 
$x^2\equiv \hat{C}\pmod{p^e}$ for $\hat{C}=ap^{\ell}(\ell<e)$, 
which is squared, and $\gcd(a,p)=1$ are given by 
$$
x = yp^{\ell/2}+kp^{e-\ell/2}, \quad k\in\mathbb{Z}[0,p^{\ell/2}-1], 
$$
where $y$ is a root of $y^2\equiv a\pmod{p^{e-\ell}}$. 
In particular, the number of its roots is $2p^{\ell/2}$.
\end{thm}
\begin{proof}
Note that $\ell$ must be even under assumption (A).
Let $x=yp^t$ such that $\gcd(y,p)=1$, and substitute $x$ into the quadratic equation; thus, we obtain:
$$
x^2 = y^2p^{2t}\equiv ap^{\ell}\pmod{p^e}.
$$
Suppose that $2t<\ell$, $y^2\equiv ap^{\ell-2t}\pmod{p^{e-2t}}$ holds, 
which means that $y$ is a multiple of $p$ and is contradictory. 
Conversely, suppose that $2t>\ell$, $y^2p^{2t-\ell}\equiv a\pmod{p^{e-\ell}}$ holds, 
which means that $a$ is a multiple of $p$ and is contradictory. 
Therefore, $2t=\ell$ holds, and we obtain:
\begin{equation}\label{yeq}
y^2\equiv a\pmod{p^{e-\ell}}.
\end{equation}
From Theorem \ref{justtwo}, (\ref{yeq}) has only two roots in modulo $p^{e-\ell}$. 
Therefore, the root $x$ of $x^2\equiv ap^{\ell}\pmod{p^e}$ can be represented as: 
\begin{equation}\label{genrep}
x = yp^{\ell/2}+bp^{e-\ell}
\end{equation}
for some $b\in\mathbb{Z}$. 
(\ref{genrep}) satisfies $x^2\equiv ap^{\ell}\pmod{p^e}$. Then, we have 
$x^2-ap^{\ell}\equiv(yp^{\ell/2}+bp^{e-\ell})^2 -ap^{\ell}= 2ybp^{e-\ell/2}+b^2p^{2e-2\ell} 
=bp^{e-\ell/2}(2y+bp^{e+\ell/2})\pmod{p^e}$. 
We learn $2by\equiv 0\pmod{p^{\ell/2}}$. Since $p$ is odd and $y$ is invertible, 
$b$ is a multiple of $p^{\ell/2}$. Then, $b=kp^{\ell/2}$ for some $k\in\mathbb{Z}$.
Substituting $b=kp^{\ell/2}$ into (\ref{genrep}), we reach 
$x = yp^{\ell/2}+kp^{e-\ell/2}$ for $k\in\mathbb{Z}[0,p^{\ell/2}-1]$. 
Since only two $y$ satisfy (\ref{yeq}), the number of $x$ is $2p^{\ell/2}$.
\end{proof}

\begin{rem}
The degenerate case occurs frequently when $\hat{N}$ has a small prime factor.
Therefore, in many cases, we can easily find the desired roots by 
brute force without using Hensel's lifting lemma.
\end{rem}

\section{Number of Candidates of Message}\label{NUM}
We really need for the perturbed public key $\hat{N}$ is to be easily factorable. 
After factoring, the problem breaks down to find the roots of the modular 
quadratic equation. 
Complexity of finding square roots mod $\hat{N}$ depends on prime factor decomposition of $\hat{N}$. 
Complicated cases arise when $\hat{N}$ is not square-free. 
Here we denote the function the number of distinct prime factor of $\hat{N}$ by $\omega(\hat{N})$, 
and the number of roots of $x^2\equiv \hat{C}\pmod{\hat{N}}$ by $\eta(\hat{C},\hat{N})$. 
CRT equality (\ref{CRT}) says that $\eta(\hat{C},\hat{N})$ is a multiplicative function. 
Combining Theorem \ref{justtwo}, \ref{zerothm}, and \ref{nonzerothm}, each $\eta(\hat{C},p^e)$ can be represented explicitly 
as follows.

\begin{thm}\label{th:numberofroots}
$$
\eta(\hat{C},\hat{N}) = \eta(\hat{C},p_1^{e_1})\cdots\eta(\hat{C},p_{\omega(\hat{N})}^{e_{\omega(\hat{N})}})
$$
for $\hat{N}=p_1^{e_1}\cdots p_{\omega(\hat{N})}^{e_{\omega(\hat{N})}}$, and where 
$$
\nu(\hat{C},p^e) = \left\{\begin{array}{lll}
2 & \mbox{if} & \gcd(\hat{C},p^e)=1 \\
2p^{\ell/2} & \mbox{if} & \gcd(\hat{C},p^e)=p^{\ell}(0<\ell<e) \\
p^{\lfloor e/2\rfloor} & \mbox{if} & \gcd(\hat{C},p^e)={p^e}.
\end{array}
\right.
$$
\end{thm}
It is well known that the following asymptotic estimate for 
$Q(x)$, which is the number of square-free numbers below $x$.

\begin{thm}(Theorem 333, Hardy-Write~\cite{Hardy-Wright}, p.355) \label{HWasym}
$$
Q(x) \sim\frac{6}{\pi^2}x+{\cal O}(\sqrt{x})\quad (x\to\infty).
$$
\end{thm}

This estimate tells us that the probability that a number should be square-free 
is approximated as $\frac{6}{\pi^2}\approx 0.6079\cdots$ for large $x$.
Note that $\ell$ is even because $\hat{C}$ is squared, and $\eta(\hat{C},p)=1$ if and only if $\gcd(\hat{C},p)=p$.
Theorem \ref{th:numberofroots} implies the following upper bound for square-free $\hat{N}$:
\begin{equation}\label{bound}
\eta(\hat{C},\hat{N})\leq 2^{\omega(\hat{N})}.
\end{equation}
The equality holds in (\ref{bound}) if and only if $\gcd(\hat{C},\hat{N})=1$. 
Therefore, Theorem \ref{th:numberofroots} and \ref{HWasym} mean that 
the probability that the inequality (\ref{bound}) holds 
is greater than or equal to $\frac{6}{\pi^2}$ asymptotically.
The asymptotic behavior of $\omega(\hat{N})$ is described by Theorem~\ref{EK}.

\begin{thm}(Erd\"os and Kac\cite{ErdosKac}) \label{EK}
The function $(\omega(n)-\ln\ln n)/\sqrt{\ln\ln n}$ is normally distributed in the 
sense that, for any fixed $z$, one has:
$$ 
\frac{1}{T}\sharp
\left\{
n\leq T;\frac{\omega(n)-\ln\ln n}{\sqrt{\ln\ln n}}\leq z\right\} \\
\to 
\int_{-\infty}^z\frac{e^{-\zeta^2/2}}{\sqrt{2\pi}}d\zeta 
$$
as $T\to\infty$, 
where we denote the cardinality of a set $A$ as $\sharp A$.
\end{thm}

Theorem \ref{EK} tells us the distribution of $\omega(\hat{N})$ obeys 
normal distribution with mean $\ln\ln \hat{N}$ and variance $\ln\ln \hat{N}$.
We use the rough estimate $\omega(\hat{N})\approx \ln\ln(2^{1024})=6.564959$ for 1,024-bit $\hat{N}$. 
However, convergence is slow because 
the term $\ln\ln \hat{N}$ increases very slowly as $\hat{N}$ becomes large.

Since both WIPR and RAMON uses at east an 80 bit challenge, 
the probability that a candidate correct message $M$ has the same data format 
accidentally is less than or equal to $1/2^{80}$.
According to the above mathematical observations, 
it is quite rare that plural candidates of $M$ will appear.


%
\section{Time of attack process}\label{SIM}
\subsection{Evaluation method}
In the attack procedure discused in Section \ref{AP}, 
the most time-consuming process is the 
factorization of $\hat{N}$ at Step 2. 
When this process takes too much time, another fault injection 
is required (Step 1). The other steps are straightforward.

 Therefore, in this section, the number of repeated fault injections and the time for factorization are 
evaluated by computer simulation. For simplicity, we explain only 
the case `crash of a byte of $N$'. 
The simulation procedure is as follows where $T_3$ is the average calculation time in Steps (3) and (4).

\begin{description}
\item[Step (0):] ~~~Generate a 1,024 bit $N$~(128 bytes). Let $N=\sum_{i=0}^{127}N[i](2^8)^i$. 
\item[Step (1a):] ~~~Decide fault position $j$.  $j = 1, 2,\ldots,126$ in uniform distribution. 
\item[Step (1b):] ~~~Decide fault pattern $k_0$. $k_0 = 1,2,\ldots,255$ in uniform distribution.
As a notation, 
$$
\hat{N}_{k_0}[i]=\left\{\begin{array}{ll}
N[i]\oplus k_0, & \mbox{if }i=j \\
N[i]. & \mbox{otherwise}
\end{array}
\right.
$$
and $\hat{N}_{k_0} = \sum_{i=0}^{127}\hat{N}_{k_0}[i](2^8)^i$.
Attackers know $j$, but do not know $k_0$.
\item[Step (2):] ~~Let $k=1$. 
Reset timer $T$, which accumulates the processing time of Steps (2a) through (4).
Reset counter $c$, which counts successful factorization.
\item[Step (2a):] ~~~~~Make $\hat{N}_k$ and attempt to factorize it within 1 min. 
If factorization is successful, proceed to Steps 3 and 4 in Section III-D. Let $c=c+1$.
If factorization does not finish, proceed to Step (4).
\item[Step (3):] ~~~If $k=k_0$, go to Step (5).
\item[Step (4):] ~~~$k=k+1$.
When $k<256$, go back to Step (2a); otherwise, proceed to Step (6).
\item[Step (5):] ~~~The attack is successful. Time consumed is $T + cT_3$. 
The simulation ends.
\item[Step (6):] ~~~The attack is in fail. Consuming time is $T + cT_3$. Simulation ends.
\end{description}

\subsection{Results}
Using a desktop PC with a Core i7-2600 CPU at 3.4 GHz with 12 GB RAM, 
the simulation runs on Mathematica 9 for Windows 7 Pro 64 bit. 
A simulated one attack begins 
at Step (0). If this attack ends at Step (5), the attack can find the correct message, i.e., the 
attack is successful. Otherwise, it ends at Step (6), which means the attack has failed. 
As a result, 28 cases were successful among 195 simulated attacks, which is an 14.4\% success rate. 
According to this result, attack with $X$ fault injections has a success rate of $1-0.856^X$. 
The success rate is 54\% for $X=5$. 
 Figure~\ref{fig:A} shows the distribution of time consumed $T$ per single attack. 
 For successful attacks, 
the mean $T$ is 115.4 min, the median is 136.8 min, and the standard deviation is 78.1 min. For 
failed attacks, the mean, median, and standard deviation of $T$ 
are 226.4 min, 226.6 min, and 4 min, respectively. $T$ can be 255 
min at most. However, when factorization finishes within 1 min, $T$ will be less than 255 min. 
Moreover, the attack is complete when the correct message is found; 
thus, $T$ becomes much less than 255 min. if the factorized fault pattern $k$ is 
the true fault pattern, $k_0$.
 The distribution of the number of successful factorizations, $c$ is shown in Fig.\ref{fig:B}. 
Here, $c$ is 30.91 
on average for failed attack cases. For successful attacks, 
the factorization will break at 
Step (3); therefore, it does not reach the final fault pattern $k=255$ with high probability. 
Thus, the successful factorization rate is estimated as $30.91/255\approx 12.1$\%.
Note that, in this simulation, the time limit for the factorization process is 
1 min at Step (2a), which is just an example. The time limit for factorization would be optimized 
by considering total attack time and/or the cost of fault injections. 
This issue will be the focus of future study.

\begin{figure}[htbp]
 \begin{center}
  \includegraphics[width=80mm]{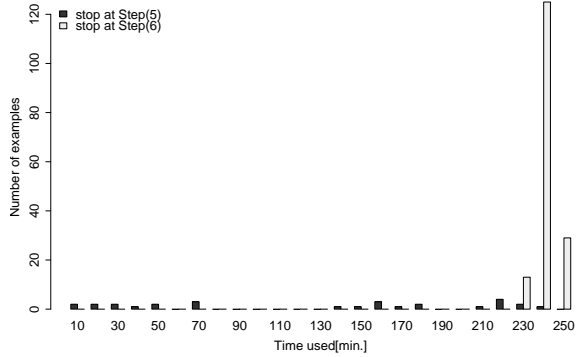}
 \end{center}
 \caption{Distribution of $T$}
 \label{fig:A}
\end{figure}

\begin{figure}[htbp]
 \begin{center}
  \includegraphics[width=80mm]{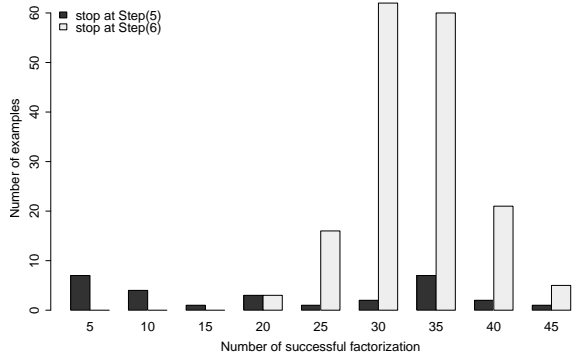}
 \end{center}
 \caption{Distribution of $c$}
 \label{fig:B}
\end{figure}

\section{Conclusion}
In this paper, we have proposed a powerful fault attack technique against a Rabin cryptosystem implemented in 
a passive RFID tag chip. 
Our attack uses one byte perturbation $\hat{N}$ of public key $N$. 
One difficulty with our attack is how to reconstruct the message $M$, including the UID, when $\gcd(M,\hat{N})\ne 1$.
We have provideed a complete algorithm to reconstruct $M$ for such cases. 
This attack requires only one fault in the public key if its perturbed public key can be factored.
The most time consuming process of our attack is the factorization of $\hat{N}$. 
Empirically, the successful factorization rate is estimated as 12.2\% of $\hat{N}$, 
even if factorization is limited within 1 min using a desktop PC. 
When this process takes too much time, another fault injection is preferable.

\section*{acknowledgements}
This work was supported by the Japan Society for the Promotion of Science KAKENHI Grant Number 25330157.



%

\end{document}